\documentclass{amsart}
\DeclareMathOperator*\Fin{Fin}
\DeclareMathOperator*\ran{ran}
\newtheorem{theorem}{Theorem}
\newtheorem{lemma}{Lemma}
\newtheorem{proposition}{Proposition}
\newtheorem{corollary}{Corollary}
\theoremstyle{definition}
\newtheorem{definition}{Definition}
\newtheorem{example}{Example}
\newcommand{\effects}{\mathcal E(\mathcal H)}
\begin{document}
\title{Extensions of witness mappings}
\author{Gejza Jen\v ca}
\address{
Department of Mathematics and Descriptive Geometry\\
Faculty of Civil Engineering\\
Slovak Technical University\\
Radlinsk\' eho 11\\
	Bratislava 813 68\\
	Slovak Republic
}
\email{gejza.jenca@stuba.sk}
\thanks{
This research is supported by grant VEGA G-1/0080/10 of M\v S SR,
Slovakia and by the Slovak Research and Development Agency under the contract
No. APVV-0071-06.
}
\subjclass{Primary: 03G12, Secondary: 06F20, 81P10} 
\keywords{effect algebra, coexistent observables, witness mappings} 
\maketitle
\begin{abstract}
We deal with the problem of coexistence in interval effect algebras
using the notion of a witness mapping.
Suppose that we are given an interval effect algebra $E$, 
a coexistent subset
$S$ of $E$, a witness mapping $\beta$ for $S$,  and an element $t\in
E\setminus S$. We study the question whether there is a witness mapping
$\beta_t$ for $S\cup\{t\}$ such that $\beta_t$ is an extension of $\beta$.
In the main result, we prove that such an extension exists if and only if there is a mapping
$e_t$ from finite subsets of $S$ to $E$ satisfying certain conditions.
The main result is then applied several times to prove 
claims of the type ``If $t$ has a such-and-such relationship to $S$
and $\beta$, then $\beta_t$ exists''.
\end{abstract}

\section{Introduction and motivation}

Let $E$ be an effect algebra. We say that a 
subset $S$ of $E$ is {\em coexistent} if and only if
there is a Boolean elgebra $B$ and a morphism of effect algebras
$\phi:B\to E$ (an {\em observable}) 
such that $S$ is a subset of the range of $E$. 

Every orthomodular lattice is an effect algebra.
It is obvious that a subset of an orthomodular lattice
is coexistent if and only if it is a subset of a block. 
This fact has a nice
generalization for lattice-ordered effect algebras: a subset of
a lattice ordered effect algebra is coexistent if and only if
it is a subset of a maximal MV-subalgebra \cite{Rie:AGoBfLEA,
Jen:BARGbMVEA}.

From the point of view of physics, the notion of coexistence
is well motivated by its application in mathematical
foundations of quantum mechanics, see for example \cite{Lud:FoQM} and
\cite{BusLahMit:TQToM}. From the purely mathematical
point of view, one can hope that the study of coexistent
subsets can shed at least some light at the enigmatic
structure of general effect algebras.

To deal with the notion of coexistence, in \cite{Jen:CiIEA} we introduced and
studied a new notion called {\em witness mapping}. For a subset $S$ of an
interval effect algebra $E$ in a partially ordered abelian group $G$, a {\em
witness mapping for $S$} is a mapping from the finite subsets of $S$ to $E$
satisfying certain conditions.

The most important result about witness mappings is the following
theorem.
\begin{theorem} (Theorem 3 of \cite{Jen:CiIEA})
Let $E$ be an interval effect algebra. $S\subseteq E$ admits a witness
mapping if and only if $S$ is coexistent.
\end{theorem}

That means that the existence of a witness mapping for $S$ is equivalent to the
existence of an observable such that $S$ is a subset of its range. However, the
definition of a witness mappings is given purely in terms of interval effect
algebras.  Moreover, despite of the fact that the proof of Theorem 1 in
\cite{Jen:CiIEA} is constructive in both directions, the relationship between
observables and witness maps is far from an one-to-one correspondence: in
general, an observable gives a rise to many different witness maps (see the
proof of Proposition 6 of \cite{Jen:CiIEA} in the context of Corollaries 3 and
4 of that paper). 

These facts give us the hope that using witness mappings we can deal
with questions concerning coexistence without explicitly dealing with
observables.

In the present paper, we continue our study of witness maps in a natural
direction. Suppose that we are given an interval effect algebra $E$, 
a coexistent subset
$S$ of $E$ a witness mapping $\beta:\Fin(S)\to E$ and an element $t\in
E\setminus S$. We want to know whether is it possible to find a witness mapping
$\beta_t:\Fin(S\cup\{t\})\to E$ such that $\beta_t$ is an extension of $\beta$.
The question of the existence of such an extension is settled by Theorem
\ref{thm:main}. We prove that $\beta_t$ exists if and only if there is a
mapping $e_t:\Fin(S)\to E$ with a certain relationship to $\beta$.  In the
remaining part of this paper, Theorem \ref{thm:main} is applied several times
to prove claims of the type ``If $t$ has a such-and-such relationship to $S$
and $\beta$, then $\beta_t$ exists''.

\section{Definitions and basic relationships}
\subsection{Effect algebras}

An {\em effect algebra} 
is a partial algebra $(E;\oplus,0,1)$ with a binary 
partial operation $\oplus$ and two nullary operations $0,1$ satisfying
the following conditions.
\begin{enumerate}
\item[(E1)]If $a\oplus b$ is defined, then $b\oplus a$ is defined and
		$a\oplus b=b\oplus a$.
\item[(E2)]If $a\oplus b$ and $(a\oplus b)\oplus c$ are defined, then
		$b\oplus c$ and $a\oplus(b\oplus c)$ are defined and
		$(a\oplus b)\oplus c=a\oplus(b\oplus c)$.
\item[(E3)]For every $a\in E$ there is a unique $a'\in E$ such that
		$a\oplus a'$ exists and $a\oplus a'=1$.
\item[(E4)]If $a\oplus 1$ is defined, then $a=0$.
\end{enumerate}

Effect algebras were introduced by Foulis and Bennett in their paper 
\cite{FouBen:EAaUQL}.

In their paper~\cite{KopCho:DP}, Chovanec and K\^ opka introduced
an essentially equivalent structure called {\em D-poset}. Their definition
is an abstract algebraic version the {\em D-poset of fuzzy sets},
introduced by K\^ opka in the paper~\cite{Kop:DPFS}.

Another equivalent structure was introduced by Giuntini and
Greuling in~\cite{GiuGre:TaFLfUP}. We refer to~\cite{DvuPul:NTiQS} for more
information on effect algebras and related topics.

\subsection{Properties of effect algebras}

In an effect algebra $E$, we write $a\leq b$ if and only if there is
$c\in E$ such that $a\oplus c=b$.  It is easy to check that for every effect
algebra $E$, $\leq$ is a partial order on $E$.  Moreover, it is possible to introduce
a new partial operation $\ominus$; $b\ominus a$ is defined if and only if $a\leq
b$ and then $a\oplus(b\ominus a)=b$.  It can be proved that, in an effect
algebra, $a\oplus b$ is defined if and only if $a\leq b'$ if and only if $b\leq
a'$. In an effect algebra, we write $a\perp b$ if and only if $a\oplus b$ exists.

A finite family $(a_1,\dots,a_n)$ of elements of an effect algebra is called
{\em orthogonal} if and only if the sum $a_1\oplus\dots\oplus a_n$ exists. An orthogonal
family $(a_1,\dots,a_n)$ is a {\em decomposition of unit} if and only if $a_1\oplus\dots a_n=1$.

\subsection{Classes of effect algebras}

The class of effect algebras can be considered
a common superclass of several important classes of algebras: orthomodular lattices
\cite{Kal:OL,Ber:OLaAA}, orthoalgebras \cite{FouRan:OS1BC}, MV-algebras \cite{Cha:AAoMVL,Mun:IoAFCSAiLSC}. 

\begin{itemize}
\item An effect algebra $E$ is an {\em orthomodular lattice} if and only if $E$ is lattice-ordered
and, for all $a,b\in E$,  $a\perp b\implies a\wedge b=0$. 
\item An effect algebra $E$ is an {\em MV-effect algebra} if and only if $E$ is lattice-ordered
and, for all $a,b\in E$, $a\ominus (a\wedge b)=(a\vee b)\ominus b$. By \cite{ChoKop:BDP},
there is a natural, one-to-one correspondence between MV-algebras and MV-effect algebras.
\item An effect algebra $E$ is a {\em Boolean algebra} if and only if $E$ is an 
orthomodular lattice and $E$ is an MV-effect algebra. In this case, we wave
$a\perp b$ iff $a\wedge b=0$ and then $a\oplus b=a\vee b$.
\end{itemize}

\subsection{Observables and coexistent subsets}

Let $E,F$ be effect algebras. A mapping $\phi:E\to F$ is a {\em morphism of effect algebras}
if and only if the following conditions are satisfied:
\begin{enumerate}
\item[(EM1)]$\phi(1)=1$.
\item[(EM2)]If $a,b\in E$, $a\perp b$ then $\phi(a)\perp\phi(b)$ and $\phi(a\oplus b)=\phi(a)\oplus\phi(b)$.
\end{enumerate}
We note that every morphism of effect algebras is isotone. Moreover, every 
morphism of effect algebras preserves the $0$ element, as well as the unary operation $x\mapsto x'$ and the partial binary operation $\ominus$.

A bijective morphism of effect algebras $\phi:E\to F$ such that $\phi^{-1}$ is a
morphism of effect algebras is called {\em an isomorphism of effect algebras}.

Let $B$ be a Boolean algebra and let $E$ be an effect algebra. A morphism of
effect algebras $\alpha:B\to E$ is called {\em an observable}. If $B$ is
finite, then we say that $\alpha$ is a {\em a simple observable}.

\begin{definition}
We say that a subset $S$ of an effect algebra is {\em coexistent} if there exists
a Boolean algebra $B$ and an observable $\alpha:B\to E$ such that $S\subseteq\alpha(B)$.
\end{definition}

\subsection{Interval effect algebras}

Let $(G,\leq)$ be a partially
ordered abelian group and 
$u\in G$ be a positive element.
For $0\leq a,b\leq u$, define $a\oplus b$ if and only if
$a+b\leq u$ and put $a\oplus b=a+b$.  With such a partial operation $\oplus$, the
closed interval 
$$
[0,u]_G=\{x\in G:0\leq x\leq u\}
$$ 
becomes an effect algebra $([0,u]_G,\oplus,0,u)$.  Effect
algebras which arise from partially ordered abelian groups in this way are
called {\em interval effect algebras}, see \cite{BenFou:IaSEA}.

By \cite{Mun:IoAFCSAiLSC}, every MV-effect algebra is an MV-algebra.

\subsection{Standard effect algebras}

Let $\mathbb H$ be a Hilbert space, let $\mathcal B_{sa}(\mathbb H)$ be the
set of all bounded self-adjoint operators on $\mathbb H$. For 
$A,B\in\mathcal B_{sa}(\mathbb H)$, write $A\leq B$ if and only if,
for all $x\in\mathbb H$, $\langle Ax,x\rangle\leq\langle Bx,x\rangle$.
Then $(\mathcal B_{sa}(\mathbb H),+,0)$ is a partially ordered abelian group.
The identity operator $I$ is a positive element of this group.

The prototype interval effect algebra is 
the {\em standard effect algebra}
$\mathcal E(\mathbb H)=[0,I]_{\mathcal B_{sa}(\mathbb H)}$. 
$\mathcal E(\mathbb H)$ plays an important role in the unsharp observable 
approach to the foundations of quantum mechanics, see for example
\cite{BusLahMit:TQToM}.

\subsection{Witness mappings}

Let $E$ be an interval effect algebra in a partially ordered abelian group $G$.
Let $S\subseteq E$. Let us write $\Fin(S)$ for the set of all
finite subsets of $S$. We write $I(\Fin(S))$ for the set of all comparable
elements of the poset $(Fin(S),\subseteq)$, that means,
$$
I(\Fin(S))=\{(X,Y)\in\Fin(S)\times\Fin(S):X\subseteq Y\}.
$$

For every mapping $\beta:\Fin(S)\to G$, we define a
mapping $D_\beta:I(\Fin(S))\to G$.
For $(X,A)\in I(\Fin(S))$,
the value
$D_\beta(X,A)\in G$ is given by the rule
$$
D_\beta(X,A):=\sum_{X\subseteq Z\subseteq A}(-1)^{|X|+|Z|}\beta(Z).
$$
The transform $\beta\mapsto D_\beta$ is (essentially) a M\"obius inversion
with respect to the poset $(\Fin(S),\subseteq)$; see \cite{Jen:CiIEA} for
details. When dealing with $D_\beta$, the following lemma is crucial.
\begin{lemma}(Lemma 1 of \cite{Jen:CiIEA})
\label{lemma:formal}
Let $E$ be an interval effect algebra in a partially ordered abelian group $G$.
Let $S$ be a subset of $E$, let $\beta:\Fin(S)\to G$.
For all $c\in S\setminus A$,
$$
D_\beta(X,A)=D_\beta(X,A\cup\{c\})+ D_\beta(X\cup\{c\},A\cup\{c\}).
$$
\end{lemma}

In \cite{Jen:CiIEA}, we introduced and studied the following notion:
\begin{definition}\label{def:cm}
Let $E$ be an interval effect algebra.
We say that a mapping $\beta:\Fin(S)\to E$ is a {\em witness
mapping for $S$} if and only if the following conditions are satisfied.
\begin{enumerate}
\item[(A1)]$\beta(\emptyset)=1$,
\item[(A2)]for all $c\in S$, $\beta(\{c\})=c$,
\item[(A3)]for all $(X,A)\in I(\Fin(S))$, $D_\beta(X,A)\geq 0$.
\end{enumerate}
\end{definition}

The most important result concerning witness mappings is the following
theorem.
\begin{theorem} (\cite{Jen:CiIEA}, Theorem 3)
Let $E$ be an interval effect algebra. $S\subseteq E$ admits a witness
mapping if and only if $S$ is coexistent.
\end{theorem}

There are at least two important examples of witness mappings:
\begin{example} (Corollary 2 of \cite{Jen:CiIEA})
Let $M$ be an MV-effect algebra. The mapping $\bigwedge:\Fin(M)\to M$ is
a witness mapping.
\end{example}
\begin{example} (Proposition 9 of \cite{Jen:CiIEA})
Let $S$ be a pairwise commuting subset of $\mathcal E(\mathbb H)$.

The mapping $\Pi:\Fin(S)\to \mathcal E(\mathbb H)$ 
given by
$$
\Pi(\{x_1,\dots,x_n\})=x_1.\dots.x_n.
$$
is a witness mapping.
\end{example}
Several properties of the witness mappings $\Pi$ and $\bigwedge$ generalize
nicely to all witness mappings, as the following proposition shows.
\begin{proposition} (Propositions 3 and 5 of \cite{Jen:CiIEA})
\begin{enumerate}
\item[(a)]
For all $(X,A)\in I(\Fin(S))$, $D_\beta(X,A)\leq 1$.
\item[(b)]
$\beta$ is an antitone mapping from $(\Fin(S),\subseteq)$ to $(E,\leq)$.
\item[(c)]
For all $X\in\Fin(S)$, $\beta(X)$ is a lower bound of $X$.
\item[(d)]
Suppose that $0\in S$. If $0\in X\in\Fin(S)$, then $\beta(X)=0$.
\item[(e)]
Suppose that $1\in S$.
For all $X\in\Fin(S)$, $\beta(X)=\beta(X\cup\{1\})$
\end{enumerate}
\end{proposition}

\section{Extensions of witness mappings in interval effect algebras}

Let $E$ be an interval effect algebra, let $S\subseteq E$ and let
$\beta$ be a witness mapping for $S$. We call the pair $(\beta,S)$
a {\em witness pair in $E$}. Suppose that there is another witness pair 
$(\beta^+,S^+)$ such that $S^+\supseteq S$ and 
$\beta^+$ restricted to $\Fin(S)$ is equal to $\beta$. We then
say that $(\beta^+,S^+)$ {\em extends} $(\beta,S)$, 
in symbols $(\beta^+,S^+)\sqsupseteq (\beta,S)$. If is easy to see
that $\sqsupseteq$ is a partial order on the set of all witness pairs.
By a standard use of Zorn lemma, it is easy to check that above every
witness there is a maximal witness pair.
If $t\in E\setminus S$ and there is a witness pair 
$(\beta_t,S\cup\{t\})$ that extends $(\beta,S)$,
then we say that {\em $(\beta,S)$ can be extended by $t$}.

In the remainder of this section, $E$ is an interval effect algebra
and $(\beta,S)$ is a witness pair in $E$.

\begin{theorem}\label{thm:main}
Let $t\in E\setminus S$.
The following are equivalent:
\begin{enumerate}
\item[(a)]
$(\beta,S)$ can be extended by $t$.
\item[(b)]
There is a mapping $e_t:\Fin(S)\to E$ such that
$e_t(\emptyset)=t$ and $0\leq D_{e_t}\leq D_\beta$.
\end{enumerate}
\end{theorem}
\begin{proof}
Let us prove that (a) implies (b). Suppose that  $(\beta,S)$ can be extended by
$t$.  Let $(\beta_t,S\cup\{t\})$ be a witness pair that extends $(\beta,S)$.
For $Y\in\Fin(S)$, put $e_t(Y)={\beta_t}(Y\cup\{t\})$. Clearly,
$e_t(\emptyset)=t$.  For all $X,A\in\Fin(S)$ with $X\subseteq A$,
\begin{align*}
D_{e_t}(X,A)=\sum_{X\subseteq Y\subseteq A}(-1)^{|X|+|Y|}e_t(Y)=\\
=\sum_{X\subseteq Y\subseteq A}(-1)^{|X|+|Y|}{\beta_t}(Y\cup\{t\})
\end{align*}
A substitution $Z=Y\cup\{t\}$ yields
\begin{align*}
\sum_{X\subseteq Y\subseteq A}(-1)^{|X|+|Y|}{\beta_t}(Y\cup\{t\})=\\
=\sum_{X\cup\{t\}\subseteq Z\subseteq A\cup\{t\}}(-1)^{|X|+|Z\cup\{t\}|}{\beta_t}(Z)=\\
=\sum_{X\cup\{t\}\subseteq Z\subseteq A\cup\{t\}}(-1)^{|X\cup\{t\}|+|Z|}{\beta_t}(Z)=\\
=D_{\beta_t}(X\cup\{t\},A\cup\{t\})\geq 0,
\end{align*}
since $\beta_t$ is a witness mapping.
Therefore, $D_{e_t}(X,A)\geq 0$.
By Lemma \ref{lemma:formal},
$$
D_{\beta_t}(X,A)=D_{\beta_t}(X,A\cup\{t\})+D_{\beta_t}(X\cup\{t\},A\cup\{t\}).
$$
As $\beta_t$ is a witness mapping,  $D_{\beta_t}(X,A\cup\{t\})\geq 0$.
Thus,
$$
D_{e_t}(X,A)=D_{\beta_t}(X\cup\{t\},A\cup\{t\})\leq D_{\beta_t}(X,A).
$$

Let us prove that (b) implies (a). The mapping $\beta_t:\Fin(S\cup\{t\})\to E$ is
given by
$$
\beta_t(X)=
\begin{cases}
\beta(X) &\text{ for }t\not\in X,\\
e_t(X\setminus\{t\}) &\text{ for }t\in X.
\end{cases}
$$
Obviously, $\beta_t(\emptyset)\beta(\emptyset)=1$ and, for all $c\in S\cup\{t\}$, $\beta_t(\{c\})=c$.

It remains to prove that $D_{\beta_t}\geq 0$. Let $X,A\in\Fin(S\cup\{t\})$ be such that
$X\subseteq A$. 

\noindent(Case 1) $t\notin A$.

If $t\notin A$, then $t\notin X$ and $D_{\beta_t}(X,A)=D_\beta(X,A)\geq 0$.

\noindent(Case 2) $t\in X$.

If $t\in X$, then $t\in A$ and 
\begin{align*}
D_{\beta_t}(X,A)=\sum_{X\subseteq Z\subseteq A}(-1)^{|X|+|Z|}\beta_t(Z)=\\
=\sum_{X\subseteq Z\subseteq A}(-1)^{|X|+|Z|}e_t(Z\setminus\{t\})
\end{align*}
A substitution $Y=Z\setminus\{t\}$ yields
\begin{align*}
\sum_{X\subseteq Z\subseteq A}(-1)^{|X|+|Z|}e_t(Z\setminus\{t\})=\\
=\sum_{X\setminus\{t\}\subseteq Y\subseteq A\setminus\{t\}}(-1)^{|X|+|Y\cup\{t\}|}e_t(Y)=\\
=\sum_{X\setminus\{t\}\subseteq Y\subseteq A\setminus\{t\}}(-1)^{|X\setminus\{t\}|+|Y|}e_t(Y)=
D_{e_t}(X\setminus\{t\},A\setminus\{t\})\geq 0,
\end{align*}
by assumption.

\noindent(Case 3) $t\not\in X$ and $t\in A$.

Put $A_0=A\setminus\{t\}$ so that
$A=A_0\cup\{t\}$ and $D_{\beta_t}(X,A)=D_{\beta_t}(X,A_0\cup\{t\})$.
By Lemma \ref{lemma:formal},
$$
D_{\beta_t}(X,A_0)=D_{\beta_t}(X,A_0\cup\{t\})+D_{\beta_t}(X\cup\{t\},A_0\cup\{t\}),
$$
hence
$$
D_{\beta_t}(X,A_0\cup\{t\})=D_{\beta_t}(X,A_0)
	-D_{\beta_t}(X\cup\{t\},A_0\cup\{t\}).$$
Since $t\notin X$ and $t\notin A_0$, we see that 
$D_{\beta_t}(X\cup\{t\},A_0\cup\{t\})=D_{e_t}(X,A_0)$
by the proof of (Case 2).
Hence,
$$
D_{\beta_t}(X,A_0\cup\{t\})=D_{\beta_t}(X,A_0)
	-D_{\beta_t}(X\cup\{t\},A_0\cup\{t\})=
D_{\beta_t}(X,A_0)-D_{e_t}(X,A_0).
$$
As $D_{\beta_t}\geq D_{e_t}$, this implies that
$D_{\beta_t}(X,A_0\cup\{t\})\geq 0$. It remains to recall that
$D_{\beta_t}(X,A)=D_{\beta_t}(X,A_0\cup\{t\})$.
\end{proof}
\begin{proposition}\label{prop:zeroone}
Every witness pair $(\beta,S)$ can be extended by $0$ and $1$.
\end{proposition}
\begin{proof}
Put $e_0(X)=0$, $e_1(X)=\beta(X)$ and apply Theorem \ref{thm:main}.
\end{proof}
\begin{proposition}
Suppose $u\in S$, $u'\notin S$. Then $(\beta,S)$ can be extended
by $u'$.
\end{proposition}
\begin{proof}
We shall apply Theorem \ref{thm:main}.
Put
$$
e_{u'}(X)=\beta(X)-\beta(X\cup\{u\}).
$$
We see that
$$
e_{u'}(\emptyset)=\beta(\emptyset)-\beta(\{u\})=1-u=u'.
$$
Let $X,A\in\Fin(S)$ be such that $X\subseteq A$. It remains to
prove that $0\leq D_{e_t}(X,A)\leq D_\beta(X,A)$.

\noindent (Case 1) $u\in X$. 

Obviously, $\beta(Z)=\beta(Z\cup\{u\})$, so 
$D_{e_{u'}}(X,A)=0$. 

\noindent (Case 2) $u\not\in A$

Let us rewrite
\begin{align*}
D_{e_{u'}}(X,A)=&\sum_{X\subseteq Z\subseteq A}(-1)^{|X|+|Z|}\beta(Z)-\beta(Z\cup\{u\})=\\
=&\sum_{X\subseteq Z\subseteq A}(-1)^{|X|+|Z|}\beta(Z)-
\sum_{X\subseteq Z\subseteq A}(-1)^{|X|+|Z|}\beta(Z\cup\{u\})=\\
=&D_\beta(X,A)-
\sum_{X\cup\{u\}\subseteq Y\subseteq A\cup\{u\}}(-1)^{|X|+|Y\setminus\{u\}|}\beta(Y)
\end{align*}
A substitution $Y=Z\cup\{u\}$ now yields
\begin{align*}
~&D_\beta(X,A)-
\sum_{X\cup\{u\}\subseteq Y\subseteq A\cup\{u\}}(-1)^{|X|+|Y\setminus\{u\}|}\beta(Y)=\\
=&D_\beta(X,A)-
\sum_{X\cup\{u\}\subseteq Y\subseteq A\cup\{u\}}(-1)^{|X\cup\{u\}|+|Y|}\beta(Y)=\\
=&D_\beta(X,A)-D_\beta(X\cup\{u\},A\cup\{u\})=D_\beta(X,A\cup\{u\}).
\end{align*}
Therefore, $0\leq D_{e_{u'}}(X,A)\leq D_\beta(X,A)$.

\noindent(Case 3) $u\in A$ and $u\not\in X$. 

Let us put $A_0=A\setminus \{u\}$.
By Lemma \ref{lemma:formal},
\begin{align*}
D_{e_{u'}}(X,A)=D_{e_{u'}}(X,A_0\cup\{u\})=
D_{e_{u'}}(X,A_0)-D_{e_{u'}}(X\cup\{u\},A_0\cup\{u\}).
\end{align*}
By (Case 1),
$$
D_{e_{u'}}(X\cup\{u\},A_0\cup\{u\})=0,
$$
so
$$
D_{e_{u'}}(X,A)=D_{e_{u'}}(X,A_0),
$$
and, since $u\not\in A_0$, (Case 3) reduces to (Case 2).
\end{proof}
\begin{proposition}
Suppose that $u\in\ran(\beta)$. Then $(\beta,S)$ can be extended by $u$.
\end{proposition}
\begin{proof}
Since $u\in\ran(\beta))$, there is $U\in\Fin(S)$ such that $\beta(U)=u$. 
For all $Z\in\Fin(S)$, we put
$$
e_u(Z)=\beta(Z\cup U).
$$
Clearly, $e_u(\emptyset)=\beta(U)=u$.
Further,
$$
D_{e_u}(X,A)=\sum_{X\subseteq Z\subseteq A}(-1)^{|X|+|Z|}\beta(Z\cup U)
$$
Note that, for every $Z$, there is a unique pair of sets $(Z_1,Z_2)$ such
that
\begin{align*}
&Z=Z_1\cup Z_2\\
&X\setminus U\subseteq Z_1\subseteq A\setminus U\\
&X\cap U\subseteq Z_2\subseteq A\cap U.
\end{align*}
In fact, $Z_1=Z\setminus U$ and $Z_2=Z\cap U$.
Therefore, we may express $D_{e_u}(X,A)$ as a double sum:
$$
D_{e_u}(X,A)=
	\sum_{X\setminus U\subseteq Z_1\subseteq A\setminus U}
	\sum_{X\cap U\subseteq Z_2\subseteq A\cap U}
		(-1)^{|Z_1|+|Z_2|+|X|}\beta(Z_1\cup Z_2\cup U).
$$
As $Z_2\subseteq U$, we obtain
\begin{equation}\label{eq:doublesum}
D_{e_u}(X,A)=
	\sum_{X\setminus U\subseteq Z_1\subseteq A\setminus U}
	\sum_{X\cap U\subseteq Z_2\subseteq A\cap U}
		(-1)^{|Z_1|+|Z_2|+|X|}\beta(Z_1\cup U).
\end{equation}

(Case 1) Suppose that $X\cap U=A\cap U$. Then 
$Z_2=X\cap U$ and the inner sum collapses to a single summand, so
$$
D_{e_u}(X,A)=
	\sum_{X\setminus U\subseteq Z_1\subseteq A\setminus U}
		(-1)^{|Z_1|+|X\cap U|+|X|}\beta(Z_1\cup U).
$$
We can substitute $Y:=Z_1\cup U$, so that the sum can be written
as
$$
D_{e_u}(X,A)=
	\sum_{X\cup U\subseteq Y\subseteq A\cup U}
		(-1)^{|Y\setminus U|+|X\cap U|+|X|}\beta(Y).
$$
By a simple reasoning about parity, it is easy to check that
$(-1)^{|X\cup U|+|U|}=(-1)^{|X\cap U|+|X|}$, hence
\begin{align*}
D_{e_u}(X,A)=&
	\sum_{X\cup U\subseteq Y\subseteq A\cup U}
		(-1)^{|Y\setminus U|+|X\cap U|+|X|}\beta(Y)=\\
	=&\sum_{X\cup U\subseteq Y\subseteq A\cup U}
		(-1)^{|Y\setminus U|+|X\cup U|+|U|}\beta(Y)=\\
	=&\sum_{X\cup U\subseteq Y\subseteq A\cup U}
		(-1)^{|Y|+|X\cup U|}\beta(Y)=D_{\beta}(X\cup U,A\cup U).
\end{align*}
By the Lemma 3 of \cite{Jen:CiIEA}, for any $C\in\Fin(S)$ with $C\cap A=\emptyset$,
$$
\bigoplus_{Y\subseteq C}D_\beta(X\cup Y,A\cup C)=D_\beta(X,A).
$$
This implies that $D_\beta(X\cup C,A\cup C)\leq D_\beta(X,A)$.
Putting $C:=U\setminus A$ yields
$$
D_\beta(X\cup(U\setminus A),A\cup(U\setminus A))\leq D_\beta(X,A).
$$
Since $X\cap U=A\cap U$, $X\cup(U\setminus A)=X\cup U$ and
$$
D_\beta(X\cup(U\setminus A),A\cup(U\setminus A))=
D_\beta(X\cup U,A\cup U).
$$
Hence,
$$
D_\beta(X\cup U,A\cup U)
\leq D_\beta(X,A).
$$
Summarizing, we see that
$$
0\leq D_\beta(X\cup U,A\cup U)=D_{e_u}(X,A)\leq D_\beta(X,A).
$$

(Case 2) Suppose that $X\cap U\neq A\cap U$.
Let us focus onto the inner sum of~(\ref{eq:doublesum}).
Pick $c\in (A\cap U)\setminus(X\cap U)$.
Consider the systems of sets
\begin{align*}
\mathbb H_c:=&\{H:
(X\cap U)\cup\{c\}\subseteq H\subseteq A\cup U\}\\
\mathbb H'_c:=&\{H:
(X\cap U)\subseteq H\subseteq (A\cup U)\setminus\{c\}\}.
\end{align*}
Note that, for all $X\cap U\subseteq Z_2\subseteq A\cap U$, 
$c\in Z_2$ iff $Z_2\in\mathbb H_c$ and $c\notin Z_2$
iff $Z_2\in\mathbb H'_c$. Moreover,
$Z_2\mapsto Z_2\cup\{c\}$ is a bijection from
$\mathbb H'_c$ onto $\mathbb H_c$.
Therefore, we may write
\begin{align*}
~&\sum_{X\cap U\subseteq Z_2\subseteq A\cap U}
	(-1)^{|Z_1|+|Z_2|+|X|}\beta(Z_1\cup U)=\\
=&\sum_{Z_2\in\mathbb H'_c}
	(-1)^{|Z_1|+|Z_2|+|X|}\beta(Z_1\cup U)+
	(-1)^{|Z_1|+|Z_2\cup\{c\}|+|X|}\beta(Z_1\cup U).
\end{align*}
However, it is obvious that
$$
(-1)^{|Z_1|+|Z_2|+|X|}\beta(Z_1\cup U)+
(-1)^{|Z_1|+|Z_2\cup\{c\}|+|X|}\beta(Z_1\cup U)=0.
$$
Thus, $D_{e_u}(X,A)=0$ for $X\cap U\neq A\cap U$.
\end{proof}
\begin{corollary}
Let $(\beta,S)$ be a maximal witness pair.
Then $0,1\in S$, $S$ is closed with respect to 
$x\mapsto x'$, and $\beta$ maps $\Fin(S)$ onto $S$.
\end{corollary}
\begin{proof}
By the Propositions 2,3, and 4.
\end{proof}

\section{Extensions in standard effect algebras}
\begin{proposition}
Let $u\in\effects$, suppose that $u$ commutes with
every element of $\ran(\beta)$. Then $(\beta,S)$ can be extended
by $u$.
\end{proposition}
\begin{proof}
Put $e_u(X)=u.\beta(X)$. Clearly, $e_u(\emptyset)=u$ and
$$
0\leq u.D_{\beta}(X,A)=D_{e_u}(X,A)\leq D_{\beta}(X,A).
$$
\end{proof}

\begin{proposition}\label{prop:convex}
Let $u,w\in\effects$, suppose that $(\beta,S)$ can be extended by
both $u$ and $w$. Let $v$ be a convex combination of $u,w$. Then
$(\beta,S)$ can be extended by $v$.
\end{proposition}
\begin{proof}
Write $v=\theta u+(1-\theta) w$, where $\theta\in[0,1]_{\mathbb R}$.
Let $e_u$ and $e_w$ be extension mappings for $u$ and $w$, respectively.
Put
$$
e_v=\theta e_u+(1-\theta)e_w .
$$
Clearly, $e_v(\emptyset)=v$ and
$$
D_{e_v}(X,A)=\theta D_{e_u}(X,A)+(1-\theta) D_{e_w}(X,A).
$$
Since $e_u$ and $e_w$ are extension mappings,
\begin{align}
\label{theta}&0\leq D_{e_u}(X,A) \leq D_\beta(X,A)\\
\label{1mtheta}&0\leq D_{e_w}(X,A) \leq D_\beta(X,A)
\end{align}
Multiplying (\ref{theta}) by $\theta$, 
(\ref{1mtheta}) by $(1-\theta)$ and then summing up the
inequalities gives us
$$
0\leq\theta D_{e_u}(X,A)+(1-\theta) D_{e_w}(X,A)\leq D_\beta(X,A).
$$
\end{proof}

\end{document}